\documentclass[11pt,reqno]{amsart}

\usepackage{amsmath, bm}
\usepackage{amssymb}
\usepackage{a4wide}
\usepackage{bbm}
\usepackage{graphics}
\usepackage{epsfig}
\usepackage[dvipsnames,svgnames,table]{xcolor}
\usepackage{mathtools}

\usepackage{tikz}

\usepackage{multirow}
\usepackage{diagbox}

\newcommand{\bk}{\color{black}}

\usepackage{hyperref}
\hypersetup{%
  colorlinks = true,
  linkcolor  = blue,
  citecolor    = red
}
\definecolor{brightmaroon}{rgb}{0.76, 0.13, 0.28}
\definecolor{britishracinggreen}{rgb}{0.0, 0.26, 0.15}
\definecolor{cadmiumgreen}{rgb}{0.0, 0.42, 0.24}
\definecolor{darkmidnightblue}{rgb}{0.0, 0.2, 0.4}
\definecolor{darkpink}{rgb}{0.91, 0.33, 0.5}
\definecolor{teal}{rgb}{0.0, 0.5, 0.5}
\definecolor{burgundy}{rgb}{0.5, 0.0, 0.13}
\definecolor{azure}{rgb}{0.0, 0.5, 1.0}
\definecolor{alizarin}{rgb}{0.82, 0.1, 0.26}
\definecolor{carminered}{rgb}{1.0, 0.0, 0.22}
\usepackage{comment}
\usepackage{amsmath, amssymb, latexsym, amscd, amsthm,amsfonts,amstext}
\usepackage[mathscr]{eucal}
\usepackage{appendix}
\usepackage[active]{srcltx}
\include{srctex}

\newcommand{\bel}{\begin{equation} \label}
\newcommand{\ee}{\end{equation}}

\usepackage{color}

 \usepackage{mleftright}
\usepackage{pgf,tikz}
\usetikzlibrary{patterns,arrows,intersections, decorations.markings, decorations.pathmorphing,
	backgrounds, positioning, fit, shapes.geometric}
\usepackage{caption}  
\usepackage{mathrsfs}

\tikzset{arrow data/.style 2 args={%
      decoration={%
         markings,
         mark=at position #1 with \arrow{#2}},
         postaction=decorate}
      }%

\usepackage{siunitx}

\usepackage{subcaption}

\newcounter{hypo}

\definecolor{gr}{rgb}  {0.,  0.69,  0.23 }
\definecolor{bl}{rgb}  {0.,  0.5,  1. }
\definecolor{mg}{rgb}  {0.85, 0.,  0.85}
\definecolor{or}{rgb}  {0.9, 0.5,  0.}

\definecolor{webred}{rgb}{0.75,0,0}
\definecolor{webgreen}{rgb}{0,0.75,0}

\newcommand{\Z}{\mathbb{Z}}

\newcommand{\C}{\mathbb{C}}

\DeclareRobustCommand{\rchi}{{\mathpalette\irchi\relax}}
\newcommand{\irchi}[2]{\raisebox{\depth}{$#1\chi$}} 

\newcommand{\boundellipse}[3]
{(#1) ellipse (#2 and #3)
}

\makeatletter
 \@addtoreset{equation}{section}
 \makeatother
 
 \usepackage{hyperref} 

\newtheorem{theorem}{Theorem}[section]
\newtheorem{lemma}[theorem]{Lemma}
\newtheorem{prop}[theorem]{Proposition}
\newtheorem{definition}[theorem]{Definition}
\newtheorem{remark}[theorem]{Remark}
\newtheorem{corollary}[theorem]{Corollary}

\numberwithin{equation}{section}
 
  {\setlength{\baselineskip}{1.5\baselineskip}

\newtheorem{df}{Definition}[section]

\newtheorem{lem}{Lemma}[section]

\newtheorem{As}{Assumption}[section]

\title[Clusters of Resonances near spectral thresholds]{Clusters of resonances for a non-selfadjoint multichannel discrete Schr\"odinger operator}
\author[M. Assal, O. Bourget, P. Miranda and D. Sambou]{Marouane Assal$^{1, \ast}\footnote{$^\ast$ Corresponding author. E-mail: marouane.assal@usach.cl}$, Olivier Bourget$^2$, Pablo Miranda$^1$ and Diomba Sambou$^3$}

\begin{document}

\date{\today}

\maketitle

\begin{quote}
\begin{itemize}
\item[$^1$] Departamento de Matem\'atica y Ciencia de la Computati\'on, \\
Universidad de Santiago de Chile, Las sophoras 173, Santiago, Chile.

\smallskip

\emph{E-mails: marouane.assal@usach.cl, pablo.miranda.r@usach.cl}

\medskip

\item[$^2$] Facultad de Matem\'aticas, Pontificia Universidad Cat\'olica de Chile,\\
Av. Vicu\~na Mackenna 4860, Santiago, Chile.

\smallskip

\emph{E-mail: bourget@uc.cl}

\medskip

\item[$^3$] Institut Denis Poisson, Universit\'e d'Orl\'eans, UMR CNRS 7013, \\
45067 Orl\'eans cedex 2, France.

\smallskip

\emph{E-mail: diomba.sambou@univ-orleans.fr}
\end{itemize}
\end{quote}

\begin{abstract} 
We study the distribution of resonances for discrete Hamiltonians of the form $H_0+V$ near the thresholds of the spectrum of
$H_0$. Here, the unperturbed operator $H_0$ is a multichannel Laplace type operator on $\ell^2(\mathbb Z; \mathbb C^N) \cong \ell^2(\mathbb Z)\otimes \mathbb C^N$ 
and $V$ is a non-selfadjoint compact perturbation. We  
compute the exact number of resonances and give a precise description on their location in clusters around some special points in the complex plane.

\vspace{0.2cm}

 
\end{abstract}

\bk 

\section{Introduction}

The resonance phenomena in quantum mechanics has been mathematically tackled from various perspectives. For instance,  resonances of a stationary quantum system can be defined as the poles of a suitable meromorphic extension of either the Green function, the resolvent of the Hamiltonian or the scattering matrix. The imaginary part of such poles is sometimes interpreted as the inverse of the lifetime of some associated quasi-eigenstate. A related idea is to identify the resonances of a quantum Hamiltonian with the discrete eigenvalues of some non-selfadjoint operator obtained from the original one by the methods of spectral deformations.
Another point of view consists in defining the resonances dynamically, i.e., in terms of quasi-exponential decay for the time evolution of the system. This property is somewhat encoded in the concept of sojourn time.  The equivalence between these different perspectives and formalisms is also an issue. 

Significant progress has been made over the last thirty years in studying the existence and asymptotic behavior of resonances, particularly in continuous configuration spaces.  This has been achieved thanks to the development of many mathematical approaches such as scattering methods, spectral and variational techniques, semiclassical and microlocal analysis (many references to this vast literature can be found in the monographs \cite{Helffer1986,hislop2012introduction,MR3969938}).

On the other side, the qualitative spectral properties of the discrete Laplace operator and some selfadjoint generalizations exhibiting dispersive properties, have been extensively investigated. We primarily refer to \cite{de1999spectral, tadano2019long} for the multidimensional lattice case ${\mathbb Z}^d$, \cite{georgescu2005isometries} and references therein for trees, and \cite{ando2016spectral, parra2018spectral} for periodic graphs and perturbed graphs, respectively.  The role of the thresholds of the discrete Laplace operator are specifically studied in \cite{Ito2013ACC,  MR3959724}.  We also refer the reader to \cite{sahbani2016spectral} and references therein for studies concerning Jacobi matrices and block Jacobi matrices.

The study of resonances for quantum Hamiltonians on discrete structures has been mainly performed on quasi 1D models. However, in these approaches, perturbations are assumed to be diagonal and essentially compactly supported.  One may quote \cite{BST} where results on the distribution of resonances of compact perturbations of the 1D discrete Laplace operator were obtained. 
This suggests that a more systematic analysis of resonances in the spirit of  
\cite{BoBrRa14_02, BonBruRai14} should be performed in this context.  In what follows, we focus on some  generalizations of the 1D discrete Laplace operator and study the distribution of resonances that appear in the neighborhoods of the thresholds, in perturbative regimes.

The asymptotic behavior  of resonances near thresholds  have been studied in an abstract setting in \cite{GriKlo96}. However, this study does not include the models of the present work.  On the other side, in some continuous waveguides models, the singularities at the thresholds are similar in structure to the ones that appear in our case. In this context, related results to ours were  obtained by one of the authors and his collaborators in   \cite{BruMirPof18, BruMirPaPo20}.  In this work,  using a  modified  approach based on elementary perturbation theory,  the  conclusions that we obtain are more general and sharper. In fact, firstly we are able to consider non-selfadjoint perturbations.  Secondly,  we provide the exact   number of resonances near each threshold, contrary to \cite{BruMirPof18, BruMirPaPo20} where only an upper bound was given, and we show that they are distributed in small clusters whose radii depend on a perturbation parameter.   

It is worth mentioning that our method does not depend on the specific structure of the Hamiltonian but rather on the nature of the singularity of the resolvent at the threshold. To be more specific, our results  can be extended to Hamiltonians whose resolvent have a singularity of the form $z^{-\frac12}$.

Consider the operator $H_0$ acting on the Hilbert space $\mathcal{H}=\ell^2(\mathbb Z)\otimes \mathbb C^N$, $N\geq1$,  defined by
 \begin{equation}\label{Model}
H_0:= \Delta \otimes I_{N} + I_{\ell^2(\mathbb Z)} \otimes M,
\end{equation}
where $\Delta$ is the positive 1D discrete Laplacian defined on $\ell^2(\mathbb Z)$ by
\begin{equation*}\label{DL}
(\Delta u)(n) := 2u(n)-u(n+1)-u(n-1), \;\;\; u\in \ell^2(\mathbb Z),
\end{equation*}
and $M$ is a $N\times N$ diagonalizable matrix. Here $I_{N}$ and $I_{\ell^2(\mathbb Z)}$ denote the identity operators on $\mathbb C^N$ and $\ell^2(\mathbb Z)$, respectively. It is well known  that the operator $\Delta$ is bounded, selfadjoint in $\ell^2(\mathbb Z)$ and its spectrum is absolutely continuous given by $\sigma(\Delta) = \sigma_{{\rm ac}}(\Delta) = [0,4]$. 
Then, the spectrum of $H_0$ is absolutely continuous and it has the following band structure in the complex plane 
$$
\sigma(H_0) = \sigma_{ac}(H_0) =\bigcup_{q=1}^N [\lambda_q, \lambda_q+4],
$$
where $\{\lambda_q\}_{q=1}^N$ are the eigenvalues of $M$.   
The endpoints $\{\lambda_q, \lambda_q+ 4\}_{q=1}^N$ are the thresholds in $\sigma(H_0)$. Let $d\in \{1,...,N\}$ be the number of distinct eigenvalues of $M$ and we set  {$\mathcal{T} =\{\lambda_q\}_{q=1}^d$}. 
For $q\in \{1,...,d\}$, denote by $\nu_q$ the dimension of ${\rm Ker}(M-\lambda_q)$, and let us denote by $\pi_q$ the projection onto ${\rm Ker}(M - \lambda_q)$ defined by 
$$
\pi_q:=\frac{1}{2\pi i}\oint_{\vert z-\lambda_q\vert=\varepsilon}(z- M)^{-1}dz, 
$$
for any $ 0<\varepsilon <{\rm min}|\lambda_q-\lambda_p|, $   $ p\neq q$.

Given a threshold $\lambda_q\in \mathcal{T}$, one introduces the parametrization
\begin{equation}\label{New variable}
k\mapsto z_q(k)= 
\lambda_q +k^2, 
\end{equation}
where $k$ is a complex variable in  
a  neighborhood of $0$. 

Let $V$ be a bounded operator in $\mathcal H$ and introduce the perturbed operator 
$$
H_V:=H_0+V.
$$

Let us  denote by $\{\delta_n\}_{n\in \mathbb Z}$ the canonical basis of $\ell^2(\mathbb Z)$ and by $\{e_j\}_{j=1}^N$ the canonical basis of $\C^N$. We will suppose the following

\begin{As}\label{H1}  There exist constants $\rho, C>0$ such that 
$$
\Vert V(n,m)\Vert_{\mathcal{B}(\C^N)}\leq C e^{-\rho(\vert n\vert+ \vert m\vert)}, \quad \forall (n,m)\in \mathbb Z^2,
$$
where $(V(n,m))_{(n,m)\in \mathbb Z^2}$ is the matrix of $V$ in the basis $(\delta_n \otimes e_j)_{(n,j)\in \mathbb Z\times \{1,...,N\}}$.\end{As}

In order to define the resonances of $H_V$  near the thresholds, let us introduce some notations. For  a separable Hilbert space $\mathcal{K}$, we  denote by 
$\mathcal{B}(\mathcal{K})$ the algebra of bounded linear operators acting on $\mathcal{K}$.  
For $s>0$, let $W_s$ be the multiplication operator by the function $\mathbb Z \ni n\mapsto \Vert e^{-\frac{s}{2}\vert \cdot \vert}\Vert_{\ell^2(\mathbb Z)} e^{\frac{s}{2}\vert n\vert}$ acting on $e^{-\frac{s}{2} \vert \cdot \vert} \ell^2(\mathbb Z)$ with values in $\ell^2(\mathbb Z)$. $W_{-s }$ stands for the multiplication operator by the function $ \mathbb Z \ni n\mapsto \|e^{-\frac{s}{2}\vert \cdot \vert}\|^{-1}_{\ell^2(\mathbb Z)} e^{-\frac{s}{2}\vert n\vert}$ acting on $\ell^2(\mathbb Z)$ with values in $\ell^2(\mathbb Z)$. We set 
$$
\bm{W}_{\pm s}:=W_{\pm s} \otimes I_{N}.$$  
For $\varepsilon>0$ and $z_0\in \mathbb C$, we set $D_{\varepsilon}(z_0):=\{z\in \mathbb C; \vert z - z_0\vert< \varepsilon\}$ and $D^*_{\varepsilon}(0):= D_{\varepsilon}(z_0)\setminus \{z_0\}$. We also define $\mathbb C_1:=\{z\in\mathbb C: {\rm Re}\,z>0, {\rm Im}\,z>0\}$.

\begin{prop}\label{MERE}  Let $\lambda_q\in \mathcal{T}$. Under Assumption \ref{H1}, there exists $ \varepsilon_0>0$ 
such that 
the operator-valued function 
$$
D^*_{\varepsilon_0}(0) \cap \mathbb C_1 \ni k\mapsto \bm{W}_{- \rho} (H_V - z_q(k))^{-1} \bm{W}_{- \rho} 
$$
admits a meromorphic extension to $D_{\varepsilon_0}(0)$, with values in $\mathcal{B}(\mathcal{H})$. We denote by 
$\mathcal{R}^{(q)}( k )$ this extension. \end{prop}
The proof of this proposition is postponed to section \ref{Prv:exten-mero}.  

\begin{definition}\label{MAINDEF}
The resonances of the operator $H_V$ near a threshold $\lambda_q\in \mathcal{T}$ are defined as the points $z_q(k)=\lambda_q +k^2$ such that $k\in D_{\varepsilon_0}(0)$ is a pole
of the meromorphic extension $\mathcal{R}^{(q)}$ 
given by Proposition \ref{MERE}. The multiplicity of a resonance $z_q(k_0)$ is defined by
\begin{equation}\label{multip}
{\rm mult}(z_q(k_0)):={\rm rank}\oint_\gamma \mathcal{R}^{(q)}(k) dk,
\end{equation}
where $\gamma$ is a positively oriented circle centered on $k_0$,  that doesn't contain any other pole of $\mathcal{R}^{(q)}$. The set of resonances of $H_V$ will be denoted by ${\rm Res}(H_V)$.
\end{definition}

\section{Main result}\label{MAINR}


Let $a_{-1}$ and $b_{-1}$ be the operators in $\ell^2(\mathbb Z)$ defined by
\begin{align}\label{29jan22}
(a_{-1} u)(n) &:= \sum_{m\in \mathbb Z} \frac{i}{2}W_{-\rho}(n) W_{-\rho}(m) u(m),\\
(b_{-1}u)(n) &:= \sum_{m\in \mathbb Z} \frac{(-1)^{n+m+1}}{2}  W_{-\rho}(n) W_{-\rho}(m) u(m).
\end{align}

For $q\neq p\in \{0,1,...,d\}$, define the projections in $\mathcal{H}$
\begin{equation}\label{PRJ}
\Pi_q := \frac{2}{i} a_{-1} \otimes \pi_q,\quad \Pi_{q,p} := \frac{2}{i} a_{-1} \otimes \pi_q-{2}b_{-1} \otimes \pi_p.
\end{equation}
Notice that ${\rm rank}\, \Pi_q = \nu_q$ and ${\rm rank}\, \Pi_{q,p} = \nu_q+\nu_p$. Introduce $E_q$ and $E_{q,p}$ as the operators defined in $\mathcal{H}$ by
\bel{defEq}  
{E_q}:=  \Pi_q  \bm{W}_{- \rho}V \bm{W}_{- \rho}\Pi_q\quad \text{and} \quad {E_{q,p}}:=  \Pi_{q,p}  \bm{W}_{- \rho}V \bm{W}_{- \rho}{\Pi_{q,p}}.
\ee 



For $q\neq 0$, let $\{\alpha_1^{(q)}, \alpha_2^{(q)}, \cdots, \alpha_r^{(q)}, 1\leq r \leq \nu_q\}$ be the set of distinct eigenvalues of ${E_q}\vert_{{\rm Ran}\, \Pi_q}$, each $\alpha_j^{(q)}$ of 
multiplicity $m_{q,j}$. Analogously  let $\{\beta_1^{(q)}, \beta_2^{(q)}, \cdots, \beta_r^{(q)}, 1\leq r \leq \nu_q+\nu_p\}$ be the set of distinct eigenvalues of ${E_{q,p}}\vert_{{\rm Ran}\, \Pi_{q,p}}$, 
each $\beta_j^{(q)}$ of multiplicity $m_{q,p,j}$. Of course $\sum_{j=1}^r m_{q,j}= \nu_q$ and $\sum_{j=1}^r m_{q,p,j}= \nu_q+\nu_p$. 

The following theorem is our main result. It gives the existence, the exact number and the asymptotic distribution of resonances near the thresholds of perturbations of $H_0$ of the form 
$$
H_{\omega V}:=H_0+\omega V
$$
with $\omega\in \mathbb C$ small.

\begin{theorem}\label{MTH1}
Assume  Hypothesis  \ref{H1}. Fix $q\in\{1,...,d\}$ and suppose that for  any $p\in \{1,...,d\}$ we have  $\lambda_q \neq \lambda_{p}+4$. Then, there exist $\varepsilon_0, \delta_0>0$ such that for all $\vert \omega\vert<\delta_0$, 
\begin{equation}\label{multi}
\sum_{z_q(k) \in {\rm Res}(H_{\omega V}) \cap D_{\varepsilon_0}(\lambda_q)} {\rm mult}(z_q(k))= \nu_q.
\end{equation} 
Furthermore, for any ${\alpha}^{(q)}_j\in \sigma(E_q)$ there are exactly  $m_{q,j}$ resonances $z_q(k)=\lambda_q+k^2$,  counted with multiplicities,  satisfying  \begin{equation}\label{ME}
k(\omega) = - \frac{i}{2}{\alpha_j^{(q)}} \omega + \mathcal{O}(\vert \omega\vert^{1+\frac{1}{m_{q,j}}}).
\end{equation}

Suppose now that $\lambda_q = \lambda_{p}+4$ for some $p\neq q\in \{1,...,d\}$. Then equations \eqref{multi}, \eqref{ME} hold true with  $\nu_q$, $\alpha_j^{(q)}$ and $m_{q,j}$ replaced by 
$\nu_q+\nu_p$, $\beta_j^{(q)}$ and $m_{q,p,j}$, respectively.
\end{theorem}

\begin{remark}
Equality \eqref{multi} gives the existence and the exact number of resonances of the operator $H_{\omega V}$ near the threshold $\lambda_q$ for small $|\omega|$. More precisely, it states that for $|\omega|$ small enough there are exactly $\nu_q$ resonances of $H_{\omega V}$ in a small disc around $\lambda_q$. On the other hand, equation \eqref{ME} implies that these resonances are distributed  in clusters around the points $-\frac{i}{2}\alpha_j^{(q)} \omega$ (see figure \ref{Fig:Ex1}).
\end{remark}

\begin{figure}[h]
\includegraphics[scale=.7]{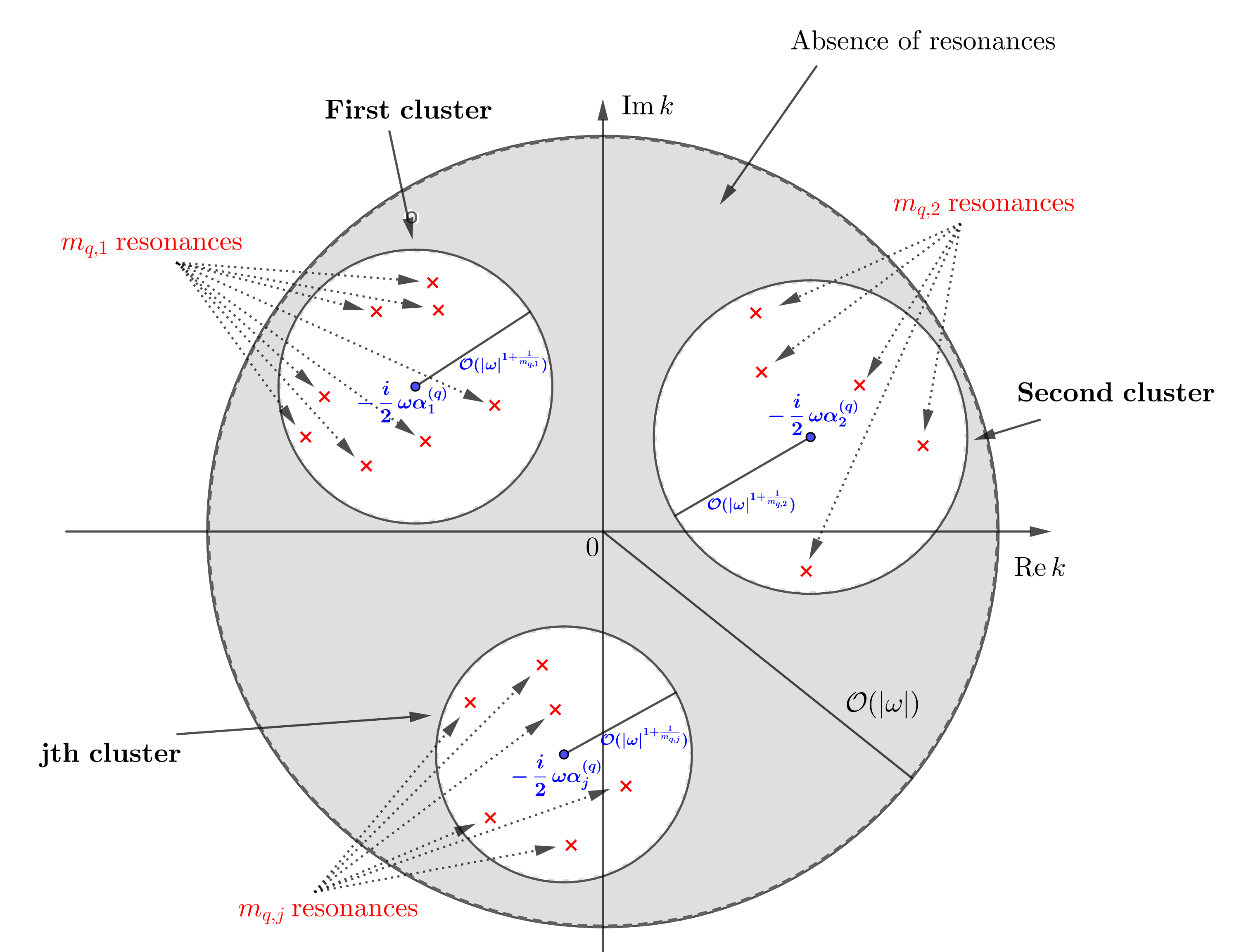}
\caption{{ \small{Resonances of $H_{\omega V}$ in variable $k$ near a  threshold $\lambda_q$} 
}}
\label{Fig:Ex1}
\end{figure}

\section{Proof of Proposition \ref{MERE}  }\label{Proofs of the results finite dimension}
\label{Prv:exten-mero}

The first step in our analysis is to study the behavior of the resolvent of the free Hamiltonian $H_{0}$ near the spectral thresholds.  Using the fact that $M$ {is diagonalizable}, 
for any $z\in \mathbb C \setminus \sigma(H_{0})$, one has 
\begin{equation}\label{RFREE}
(H_{0} - z)^{-1} = \sum_{j=1}^d (\Delta + \lambda_j -z)^{-1} \otimes \pi_j .
\end{equation}

Let us recall the following basic properties of the one-dimensional discrete Laplacian on $\ell^2(\mathbb Z)$. Let $\mathcal{F}: \ell^2(\mathbb Z)\to L^2(\mathbb T)$ be the unitary 
discrete Fourier transform defined by 
$$
(\mathcal{F}u)(\theta) = \frac{1}{2\pi}\sum_{n\in \mathbb Z} e^{-in\theta} u(n), \quad u\in \ell^2(\mathbb Z), \theta \in \mathbb T:=\mathbb R/2\pi \mathbb Z.
$$
The operator $\Delta$ is unitarily equivalent to the multiplication operator on $L^2(\mathbb T)$ by the function $f: \theta \mapsto 2-2\cos(\theta)$. More precisely, one has 
\begin{equation}\label{24nov21_b}
(\mathcal{F} (\Delta u))(\theta) = f(\theta) (\mathcal{F}u)(\theta), \quad u\in \ell^2(\mathbb Z), \theta\in \mathbb T.
\end{equation}
Hence, the operator $\Delta$ is selfadjoint in $\ell^2(\mathbb Z)$, its spectrum is absolutely continuous and coincides with the range of the function $f$, that is $\sigma(\Delta) = \sigma_{{\rm ac}}(\Delta) = [0,4]$. 

For any $z\in \mathbb C\setminus [0,4]$, the kernel of $(\Delta - z)^{-1}$ is given by (see for instance  \cite{KKK})
\begin{equation}\label{KER}
R_0(z;n,m) = \frac{e^{-i\theta(z)\vert n-m\vert}}{2 i \sin(\theta(z))}, \quad (n,m)\in \mathbb Z^2,
\end{equation}
where $\theta(z)$ is the unique solution to the equation $2  - 2\cos(\theta) = z$ lying in the region $\{\theta\in \mathbb C; -\pi \leq {\rm Re}\,\theta\leq \pi, {\rm Im}\,\theta<0\}$.

Using the previous  representation of the kernel, we have the following  result whose proof is elementary and omitted here.
\begin{lemma}\label{LBST} 
{Let $z_0\in [0,4)$ and $\rho>0$. There exists $\varepsilon_0>0 $ such that the operator-valued function 
$$
D_{\varepsilon_0}^*(0)\cap \mathbb C_1 \ni k\mapsto W_{-\rho} (\Delta - (z_0+ k^2))^{-1} W_{-\rho}
$$
admits an analytic extension to $D_{\varepsilon_0}^*(0)$ if $z_0=0$ and to $D_{\varepsilon_0}(0)$ if $z_0\in (0,4)$, with values in the Hilbert-Schmidt 
class operators in $\ell^2(\mathbb Z)$. This extension will be denoted by $R_0(z_0+k^2)$.}

\end{lemma}

In the next result, we show that the weighted resolvent of the free Hamiltonian $H_0$ extends meromorphically near any $\lambda_q\in \mathcal{T}$ and we precise the nature of its singularity at $\lambda_q$. 
\begin{lemma}\label{PLHE} 
Let $\lambda_q\in \mathcal{T}$. There exists $\varepsilon_0>0$ such that the operator-valued function 
\begin{equation}\label{MOVF}
D^*_{\varepsilon_0}(0)\cap \mathbb C_1 \ni k \mapsto \bm{W}_{-\rho} (H_{0}-z_q(k))^{-1} \bm{W}_{- \rho}
\end{equation}
admits an analytic extension to $D^*_{\varepsilon_0}(0)$, with values in $\mathfrak{S}_{\infty}(\mathcal{H})$, denoted $\mathcal{R}_{0}^{(q)}(k)$. Moreover:
\begin{itemize}
\item[(i)] If $\lambda_q $ is non-degenerate, i.e., $\lambda_q \neq \lambda_{p}+4$ for  any $p\in \{1,...,d\}$, then 
\begin{equation}
\mathcal{R}_{0}^{(q)}(k)- \frac{a_{-1} \otimes \pi_q}{k} \in {\rm Hol}\big(D_{\varepsilon_0}(0);\mathfrak{S}_{\infty}(\mathcal{H}) \big).
\end{equation}
\item[(ii)] If $\lambda_q $ is degenerate, i.e., $\lambda_q = \lambda_p +4 $ for some $p\in\{1,...,d\}$, then 
\begin{equation}
\mathcal{R}_{0}^{(q)}(k)- \frac{a_{-1} \otimes \pi_q + b_{-1}\otimes \pi_p}{k} \in {\rm Hol}\big(D_{\varepsilon_0}(0);\mathfrak{S}_{\infty}(\mathcal{H}) \big).
\end{equation}
\end{itemize}

\end{lemma}
\begin{proof}
Let us start with the proof of the first part of the result on the analytic extension. Setting $z_j^{(q)}:= \lambda_q - \lambda_j$ 
it follows from \eqref{RFREE} that
 \begin{equation}\label{FRZ} \bm{W}_{-\rho} (H_{0}-z_q(k))^{-1} \bm{W}_{- \rho}  = \sum_{j=1}^d W_{-\rho} \big(\Delta - (z_j^{(q)} +k^2)\big)^{-1} W_{-\rho}  \otimes \pi_j.
\end{equation}
The above sum splits into the following two terms
\begin{align}\label{FRP}
\bm{W}_{-\rho} (H_{0}-z_q(k))^{-1} \bm{W}_{- \rho} = &\sum_{z_j^{(q)} \in [0,4]} W_{-\rho} \big(\Delta -( z_j^{(q)} + k^2) \big)^{-1} W_{-\rho}  \otimes \pi_j  \\
+ &\sum_{z_j^{(q)} \notin [0,4]} W_{-\rho} \big(\Delta -( z_j^{(q)} + k^2) \big)^{-1} W_{-\rho}  \otimes \pi_j. \nonumber
\end{align}
The second term in the RHS is clearly analytic with respect to $k$ in a small neighborhood of $0$. On the other hand, by Lemma \ref{LBST}, the first term in 
the RHS extends to an analytic function of $k$ in $D_{\varepsilon_0}^*(0)$ for $\varepsilon_0>0$ small enough. The extension is clearly in $\mathfrak{S}_{\infty}(\mathcal{H})$.

We turn now to the proof of the second part of the result. We only prove (i) since the proof of the other assertion works similarly. Recall from \eqref{FRZ} that we have 
\begin{align*}
\bm{W}_{-\rho} (H_{0}-z_q(k))^{-1} &\bm{W}_{- \rho} = \sum_{j\neq q} W_{-\rho} \big(\Delta -( z_j^{(q)} + k^2) \big)^{-1} W_{-\rho}  \otimes \pi_j  \\
&+ W_{-\rho} (\Delta - k^2)^{-1} W_{-\rho}  \otimes \pi_q.
\end{align*}
Since $\lambda_q \neq \lambda_j+4$ for all $j\neq q$, it follows that $z_j^{(q)} \notin \{0,4\}$ for any $j\neq q\in \{1,...,d\}$. Consequently, by Lemma \ref{LBST}, 
the first term in the RHS of the above equation extends analytically in a small neighborhood of $0$ with values in $\mathfrak{S}_\infty(\mathcal{H})$. On the other 
hand, the kernel of the operator $W_{-\rho} (\Delta - k^2)^{-1} W_{-\rho}$ is given by 
\begin{equation}\label{KAZ}
\frac{e^{-\frac{\rho}{2}\vert n\vert}}{\|e^{-\frac{\rho}{2}\vert \cdot \vert}\|_{\ell^2(\mathbb Z)}} \frac{e^{-\frac{\rho}{2}\vert m\vert}}{\|e^{-\frac{\rho}{2}\vert \cdot \vert}\|_{\ell^2(\mathbb Z)}}R_0(k;n,m) ,
\end{equation}
where $R_0(k^2;n,m)$ is defined by \eqref{KER}. One can write
\begin{equation}\label{KAZ1}
R_0(k^2;n,m) = \frac{i}{k\sqrt{4-k^2}}+ \frac{i( e^{i\vert n-m\vert 2 \arcsin \frac{k}{2}} -1 )}{k\sqrt{4-k^2}} = \frac{i}{2k}+ r(k;n,m),
\end{equation}
with
$$
r(k;n,m) := i \Big( \frac{1}{k\sqrt{4-k^2}} - \frac{1}{2k}\Big) + \frac{i(e^{i\vert n-m\vert 2 \arcsin \frac{k}{2}} -1)}{k\sqrt{4-k^2}}.
$$
One easily verifies that the function $r$ extends to a holomorphic function in a small neighborhood of $0$. Therefore, putting together \eqref{KAZ} and \eqref{KAZ1}, one obtains  
\begin{equation}
W_{-\rho} (\Delta - k^2)^{-1} W_{-\rho} \otimes \pi_q = \frac{a_{-1} \otimes \pi_q }{k} + \mathcal{A}(k) \otimes \pi_q,
\end{equation}
where $\mathcal{A}(k)$ acts on $\ell^2(\mathbb Z)$ with kernel $W_{-\rho}(n) r(k;n,m) W_{-\rho}(m)$. This ends the proof. 
\end{proof}

\subsection{Proof of Proposition \ref{MERE}}
The proof is a consequence of Lemma \ref{PLHE} and the analytic Fredholm extension Theorem. {From the resolvent identity 
$$
(H_V- z)^{-1} (I+V (H_{0} - z)^{-1}) = (H_{0} - z)^{-1},
$$
it follows that
\begin{equation}\label{pert_resolvent}
\bm{W}_{-\rho} (H_V - z_q(k))^{-1} \bm{W}_{-\rho} = \bm{W}_{-\rho} (H_{0} - z_q(k))^{-1} 
\bm{W}_{-\rho} \big( I+ \mathcal{P}(z_q(k)) \big)^{-1},
\end{equation} 
where
\begin{equation}\label{DPP}
\mathcal{P}(z) :=  \bm{W}_{\rho} V (H_{0} - z)^{-1} \bm{W}_{-\rho}.
\end{equation}
Lemma \ref{PLHE} implies that there exists $\varepsilon_0>0$ such that the operator-valued function 
$k\mapsto \mathcal{P}(z_q(k))$ defined by \eqref{DPP}  extends to an analytic function in $D^*_{\varepsilon_0}(0)$ with values in $\mathfrak{S}_{\infty}(\mathcal{H})$. Therefore the analytic Fredholm theorem ensures that 
$$
D^*_{\varepsilon_0}(0) \cap \mathbb C_1 \ni k\mapsto \big( I+ \mathcal{P}(z_q(k)) \big)^{-1}
$$
admits a meromorphic extension to $D^*_{\varepsilon_0}(0)$. We use the same notation for the extended operator. Hence, 
the operator-valued function $k\mapsto \bm{W}_{-\rho} (H_V- z_q(k))^{-1} \bm{W}_{-\rho}$ extends to a 
meromorphic function of $k\in D^*_{\varepsilon_0}(0)$. This ends the proof of Propositions \ref{MERE}.

\section{Proofs Theorems \ref{MTH1}}

Let us start with the following preliminary results which precise the nature of the singularity of the meromorphic extension of the weighted resolvent at the thresholds. Recall that $a_{-1}$ is the operator in $\ell^2(\mathbb Z)$ with kernel $\frac{i}{2}W_{-\rho}(n) W_{-\rho}(m)$ and we introduce the operator $b_{-1}$ in  $\ell^2(\mathbb Z)$ with kernel $-\frac{(-1)^{n+m}}{2}  W_{-\rho}(n) W_{-\rho}(m)$.
\begin{prop}\label{PE1} Let $\lambda_q \in \mathcal{T}$ be such that $\lambda_q \neq \zeta_p$ for all $p\neq q$. There exists $\varepsilon_0>0$ small enough such that for $k\in D^*_{\varepsilon_0}(0)$ the following statements hold:
\begin{itemize}
\item[(i)] If $\lambda_q = \lambda_q$ then 
\begin{equation}
\bm{W}_{-\rho} (H_{0}-z_q(k))^{-1} \bm{W}_{- \rho} - \frac{a_{-1} \otimes \pi_q}{k} \in \text{Hol}\,\big(D_{\varepsilon_0}(0);\mathfrak{E} \big).
\end{equation}
\item[(ii)] If $\lambda_q = \lambda_q + 4$ then 
\begin{equation}
\bm{W}_{-\rho} (H_{0}-z_q(k))^{-1} \bm{W}_{- \rho} - \frac{b_{-1} \otimes \pi_q }{k} \in \text{Hol}\,\big(D_{\varepsilon_0}(0);\mathfrak{E} \big).
\end{equation}
Here $\mathfrak{E}$ is equals to $\mathcal{S}_\infty(\ell^2(\mathbb Z)\otimes \mathfrak{G})$ in the case $\bullet = A$ and equals to $\mathcal{B}(\ell^2(\mathbb Z;\mathcal{G}))$ in the case $\bullet = B$.
\end{itemize}

\end{prop}

\begin{proof}
Let us start with the case \textbf{(A)}. Suppose that $\lambda_q = \lambda_q$ and recall from \eqref{RFREE} that we have 
\begin{align*}
\bm{W}_{-\rho} (H_{0}-z_q(k))^{-1} &\bm{W}_{- \rho} = \sum_{j\neq q} W_{-\rho} \big(\Delta -( \lambda_q - \lambda_j + k^2) \big)^{-1} W_{-\rho}  \otimes \pi_j  \\
&+ W_{-\rho} \big(\Delta - k^2 \big)^{-1} W_{-\rho}  \otimes \pi_q.
\end{align*}
Since $\lambda_q \neq \lambda_j+4$ for all $j\neq q$ it follows that $\lambda_q - \lambda_j \notin \{0,4\}$ for any $j\neq q\in \{1,...,d\}$. Consequently, by Lemma \ref{LBST}, the first term in the RHS of the above equation extends analytically in a small neighborhood of $0$. On the other hand, taking the Laurent expansion of $k\mapsto W_{-\rho} \big(\Delta - k^2 \big)^{-1} W_{-\rho}  \otimes \pi_q$ near $k=0$ using the fact that the kernel of $\big(\Delta - k^2 \big)^{-1}$ is given by according to \eqref{KER} one gets 
$$
W_{-\rho} \big(\Delta - k^2 \big)^{-1} W_{-\rho}  \otimes \pi_q - \frac{1}{k} a_{-1} \otimes \pi_q \in \text{Hol}\,\big(D_{\varepsilon_0}(0); \mathcal{S}_\infty(\ell^2(\mathbb Z)\otimes \mathfrak{G}) \big),
$$
which proves the claim of (i). The proof of (ii) is similar except that in this case the kernel of $(\Delta - (4+k^2))^{-1}$ is given by $\frac{-(-1)^{n+m}}{2} W_{-\rho}(n)W_{-\rho}(m)$.

Consider now the case $\textbf{(B)}$ and suppose that $\lambda_q = \lambda_q$. From \eqref{RFREE} we have 
\begin{align*}
\bm{W}_{-\rho} (H_{0}-z_q(k))^{-1} &\bm{W}_{- \rho} = \sum_{j=1}^d W_{-\rho} \big(\Delta -( \lambda_q - \lambda_j + k^2) \big)^{-1} W_{-\rho}  \otimes \pi_j  \\
&+ W_{-\rho} \big(\Delta - (\lambda_q +k^2) \big)^{-1} W_{-\rho}  \otimes \pi_0.
\end{align*}
Assume first that $\lambda_q \neq 0$. It follows from Lemma \ref{LBST} that the second term in the RHS of the above equation extends analytically in a small neighborhood of $0$. The first term can be treated as above and then we get the claim. The case $\lambda_q = 0$ follows again from Lemma \ref{LBST} and the above analysis. 
\end{proof}

\section{Proof of Theorem \ref{MTH1}}  
\subsection{Proof of of Theorem \ref{MTH1}} 


According to Lemma \ref{PLHE}, there exists $\varepsilon_0>0$ and an analytic function $\mathcal{G}$ in $D_{\varepsilon_0}(0)$ with values in $\mathfrak{S}_\infty(\mathcal{H})$ such that for all $k\in D_{\varepsilon_0}^*(0)$ we have 
\begin{equation}\label{TFW1}
\mathcal{R}_{0}^{(q)}(k) = \frac{a_{-1} \otimes \pi_q }{k}+ \mathcal{G}(k).
\end{equation}

It follows from equation \eqref{pert_resolvent} that for all $k\in D_{\varepsilon_0}^*(0)$, 
\begin{equation}\label{CWZ}
\mathcal{R}_{\omega}^{(q)}(k) = \Big( \frac{a_{-1} \otimes \pi_q }{k} + \mathcal{G}(k) \Big) \left[ I+ \mathcal{P}_{\omega}(z_q(k)) \right]^{-1},
\end{equation}
where $\mathcal{R}_{\omega}^{(q)}(k)$ is the meromorphic extension of $k\mapsto \bm{W}_{- \rho} (H_{\omega V} - z_q(k))^{-1} \bm{W}_{- \rho} $ given by Proposition \ref{MERE} and $\mathcal{P}_{\omega}(z_q(k))$ is defined by \eqref{DPP}. More precisely, 
setting ${V}_\rho:=  \bm{W}_\rho V \bm{W}_\rho$, 
one has 
\begin{equation}\label{TFW2}
\left[ I+ \mathcal{P}_{\omega}(z_q(k)) \right]^{-1} = \big( I+ \omega {V}_\rho \mathcal{R}_{0}^{(q)}(k) \big)^{-1}.
\end{equation}
Since $\mathcal{G}$ is analytic near $0$, it follows that for $\vert \omega\vert$ small enough, the operator-valued function $I+ \omega {V}_\rho \mathcal{G}(k)$ is invertible. 
Using \eqref{TFW1}, one writes
\begin{equation}\label{TRZ}
\left[ I+ \mathcal{P}_{\omega}(z_q(k)) \right]^{-1} =   \Big( I+\frac{\omega}{k}\mathcal{L}_\omega(k) \Big)^{-1} \left( I + \omega  {V}_\rho \mathcal{G}(k)\right)^{-1},
\end{equation}
where $\mathcal{L}_\omega(k)$ is the operator in $\mathcal{H}$ defined by
$$
\mathcal{L}_\omega(k) :=   \left( I + \omega {V}_\rho \mathcal{G}(k)\right)^{-1} {V}_\rho (a_{-1} \otimes \pi_q).
$$
Putting together \eqref{CWZ} and \eqref{TRZ}, we obtain that for all $k\in D_{\varepsilon_0}^*(0)$
\begin{equation}\label{29nov21b}
\mathcal{R}_{\omega}^{(q)}(k) = \Big( \frac{a_{-1} \otimes \pi_q }{k} + \mathcal{G}(k) \Big)   \Big( I+\frac{\omega}{k}\mathcal{L}_\omega(k) \Big)^{-1} \left( I + \omega  {V}_\rho \mathcal{G}(k)\right)^{-1}.
\end{equation}
Then, the poles of $k \mapsto\mathcal{R}_{\omega}^{(q)}(k)$ near $0$ coincide with those of the operator-valued function
$$
 k \mapsto J_\omega(k) : = \Big( \frac{a_{-1} \otimes \pi_q }{k} + \mathcal{G}(k) \Big) \Big( I+\frac{\omega}{k}\mathcal{L}_\omega(k) \Big)^{-1}. 
$$
We shall make use of the following elementary result whose proof is omitted.
\begin{lemma}\label{lem2}
Let $\mathcal{K}$ be a Hilbert space and consider two linear operators $A, \Pi : \mathcal{K} \to \mathcal{K}$ such that $\Pi^2= \Pi$ and $A\Pi = A$. Then, $I+A$ is invertible if and only if $\Pi(I + A) \Pi : {\rm Ran} \, \Pi \to {\rm Ran} \, \Pi$ is invertible, and in this case one has 
$$
(I+ A)^{-1} =  (I - \widetilde{\Pi} A \Pi) B^{-1}  + \widetilde{\Pi} ,
$$
where $\widetilde{\Pi}:= I- \Pi$ and $B^{-1} :=(\Pi (I+A)\Pi)^{-1} \oplus 0$ with respect to the decomposition $\mathcal{K} = {\rm Ran} \, \Pi \oplus {\rm Ran} \, \widetilde{\Pi}$.
\end{lemma}


Let $\Pi_q$ be the projection on $\mathcal{H}$ defined by \eqref{PRJ}. Applying the above result with $A= \frac{\omega}{k} \mathcal{L}_\omega(k)$ and $\Pi = \Pi_q$, we get
$$
\Big( I+\frac{\omega}{k}\mathcal{L}_\omega(k) \Big)^{-1} = \Big(I - \frac{\omega}{k} \widetilde{\Pi}_q \mathcal{L}_\omega(k) \Pi_q  \Big) \biggl(  \left(\Pi_q \left(I+\frac{\omega}{k}\mathcal{L}_\omega(k) \right) \Pi_q\right)^{-1} \oplus 0 \biggr) + \widetilde{\Pi}_q.
$$
Here, $\widetilde{\Pi}_q:= I - \Pi_q$. Therefore, a straightforward computation yields 
$$
J_\omega(k) = \left( \frac{i}{2} \Pi_q - \omega\mathcal{G}(k) \Big( \frac{k}{\omega} - \widetilde{\Pi}_q \mathcal{L}_\omega(k) \Pi_q \Big) \right) \big( \left[\Pi_q \left(k +\omega\mathcal{L}_\omega(k) \right) \Pi_q\right]^{-1} \oplus 0 \big)
+ \mathcal{G}(k) \widetilde{\Pi}_q.
$$
Since $\left[\Pi_q \left(k +\omega\mathcal{L}_\omega(k) \right) \Pi_q\right]^{-1} : {\rm Ran} \, \Pi_q \to  {\rm Ran} \, \Pi_q $, it follows that $\left[\Pi_q \left(k +\omega\mathcal{L}_\omega(k) \right) \Pi_q\right]^{-1} \oplus 0$ is stable by $\Pi_q : \mathcal{H} \to \mathcal{H}$. Consequently, 
\begin{equation}\label{29jan22b}
J_\omega(k) = \left( \frac{i}{2}  - \omega\mathcal{G}(k) \Big( \frac{k}{\omega} - \widetilde{\Pi}_q \mathcal{L}_\omega(k) \Pi_q \Big) \right)  \big( \left[\Pi_q \left(k +\omega\mathcal{L}_\omega(k) \right) \Pi_q\right]^{-1} \oplus 0 \big)
+ \mathcal{G}(k) \widetilde\Pi_q.
\end{equation}
Using the analyticity of $\mathcal{G}$ and $\mathcal{L}_\omega$ near $0$, one sees that $\frac{i}{2}  - \omega\mathcal{G}(k) \big( \frac{k}{\omega} - \widetilde{\Pi}_q \mathcal{L}_\omega(k) \Pi_q \big)$ 
is invertible for $\vert k\vert$ and $\vert \omega \vert$ small enough. Therefore, we conclude that the poles of $J_\omega$ near $0$ are the same to those of the operator-valued function 
\begin{equation*}
k\mapsto \left(\Pi_q \left(k+\omega\mathcal{L}_\omega(k) \right) \Pi_q\right)^{-1} : {\rm Ran} \, \Pi_q \to {\rm Ran} \, \Pi_q.
\end{equation*}
Let $M_\omega(k)$ be the matrix of the operator $\Pi_q \left(k+\omega\mathcal{L}_\omega(k) \right) \Pi_q : {\rm Ran} \, \Pi_q \to {\rm Ran} \, \Pi_q$. We have 
\begin{align} \label{AGGN}
\Pi_q \left(k +\omega\mathcal{L}_\omega(k) \right) \Pi_q &= k \Pi_q + \omega  \Pi_q \left( I + \omega {V}_\rho \mathcal{G}(k)\right)^{-1} {V}_\rho (a_{-1}\otimes \pi_q)\Pi_q \nonumber \\
&= k \Pi_q + \frac{i}{2}\omega   \Pi_q{V}_\rho\Pi_q+\omega^2 \Pi_q S_\omega(k) \Pi_q,
\end{align}
where 
\bel{23}
k \mapsto S_\omega(k):= \frac{i}{2} \sum_{n\geq 1} (-1)^n \omega^{n-1} ( {V}_\rho \mathcal{G}(k))^n {V}_{\rho}
\ee
 is an  operator-valued function which is  analytic near $k=0$ for $\vert \omega \vert>0$ small enough and $\|S_\omega(k)\|=\mathcal{O}(1)$ uniformly w.r.t. $k$. 

The usual expansion formula for the determinant allows to write
$$
{\rm det}(M_\omega(k))=  \omega^{\nu_q} \biggl( \prod_{j=1}^{r} \Big( \frac{k}{\omega}+ \boldsymbol{\alpha}_j^{(q)} \Big)^{m_{q,j}} + \omega s_\omega(k) \biggr),
$$
with $\boldsymbol{\alpha}_j^{(q)} := \frac{i}{2} \alpha_j^{(q)}$, where $\{{\alpha_j^{(q)}}\}_{j=1}^r$ are the distinct eigenvalues of $E_q=  \Pi_q{V}_\rho\Pi_q: {\rm Ran}\, \Pi_q \to {\rm Ran}\, \Pi_q$ and $s_\omega$ is an analytic scalar-valued function satisfying 
\begin{equation}\label{UTAE}
\vert s_\omega(k)\vert \leq C_0,
\end{equation} 
for some constant $C_0>0$ independent of $k$ and $\omega$. We are therefore led to study the roots of the equation 
\bel{10}
\prod_{j=1}^{r} \Big( \frac{k}{\omega}+ \boldsymbol{\alpha}_j^{(q)} \Big)^{m_{q,j}} + \omega s_\omega(k) = 0.
\ee
On the one hand, by a simple contradiction argument one shows that all the roots of the above equation satisfy \eqref{ME}. On the other hand, let $j_0\in \{1,...,r\}$ and let $C>0$ be a constant independent of $k$ and $\omega$. We set 
\begin{center}
$\delta_{j_0,C}:= C\vert \omega\vert^{1+\frac{1}{m_{q,j_0}}}.$ 
\end{center}
There exists a constant $C'>0$ such that for any $k\in \partial D_{\delta_{j_0,C}}(-\boldsymbol{\alpha}_{j_0}\omega)$, one has 
$$
 \prod_{j=1, j\neq j_0}^r \Big\vert\frac{k} \omega+\boldsymbol{\alpha}_j^{(q)}\Big\vert^{m_{q,j}} \geq C'.
$$
Consequently, 
$$
 \prod_{j=1}^{r} \Big\vert \frac{k}{\omega}+ \boldsymbol{\alpha}_j^{(q)} \Big\vert^{m_{q,j}}\geq C C'\vert \omega \vert> \vert \omega s_{\omega}(k)\vert, \quad \forall k \in \partial D_{\delta_{j_0,C}}(-\boldsymbol{\alpha}_{j_0}\omega),
$$
where $C>0$ is chosen such that $CC'>C_0$, with $C_0$ given by \eqref{UTAE}.

Since both terms in \eqref{10} are analytic functions of $k$ near $k=0$, it follows by Rouch\'e Theorem that for $\vert \omega\vert$ small, ${\rm det}(M_\omega(k))$ admits exactly $m_{q,j_0}$ zeros in $D_{\delta_{j_0,C}}(- \boldsymbol{\alpha}_{j_0}\omega)$, counting multiplicities. This ends the proof of statement \eqref{ME}.

Let us now prove \eqref{multi}.  Fix $q\in \{1,...,d\}$. Equation \eqref{ME} implies that in variable $k$, the resonances of $H_{\omega V} $ are distributed in ``clusters" around  the points $-\frac{i}2\omega\alpha_j^{(q)}$, $j\in \{1,...,r\}$. Fix $j\in \{1,...,r\}$ and let $C>0$ and $1<\delta<1+1/m_{q,j}$ so that the disk $D_{C\vert \omega \vert^\delta}(-\frac{i}2\omega\alpha_j^{(q)})$ contains all the resonances of the j-th cluster and only them. Set $\Gamma_j := \partial D_{C\vert \omega \vert^\delta}(-\frac{i}2\omega\alpha_j^{(q)})$.
We will  show that 
\begin{equation}\label{29nov21g}
{\rm rank}\oint_{\Gamma_j} \mathcal{R}^{(q)}_{\omega}(k)dk = {\rm rank}\oint_{\Gamma_j} (k\Pi_q +\frac{i}2\omega\Pi_q{V}_\rho\Pi_q)^{-1}\,dk.
\end{equation}

\begin{lem}
Define 
$$
P_{j,\omega}:=\oint_{\Gamma_j}(k+\omega E_q+w^2 \Pi_q S_\omega(k) \Pi_q)^{-1}dk;\quad \tilde{P}_{j,\omega}:=\oint_{\Gamma_j}(k+\omega E_q+\omega^2\Pi_q S_\omega(-\omega\alpha_j) \Pi_q)^{-1}dk.
$$ 
Then, $\|P_{j,\omega}-\tilde{P}_{j,\omega}\|=o(1)$ when  $|\omega|\to0$. 
\end{lem}
\begin{proof}
First, let $J_q$ be   the  canonical Jordan form of    $E_q=\mathcal{U}J_q\mathcal{U}^{-1}$. Then $$
P_{j,\omega}=\mathcal{U}\oint_{\Gamma_j}(k+\omega J_q+w^2Q_\omega(k))^{-1}dk\,\mathcal{U}^{-1}=\mathcal{U}\oint_{\gamma_j}(\zeta+ J_q+wQ_\omega( \zeta))^{-1}d\zeta\,\mathcal{U}^{-1}
$$
where $k=\omega\zeta$, $Q_\omega(\zeta)=:\mathcal{U}^{-1} \Pi_q S_\omega(\omega \zeta)  \Pi_q\mathcal{U}$ and $\gamma_j$ is the curve $\gamma_j:=\{-\alpha_j+ \vert \omega \vert^{\delta-1} e^{i(\theta - \arg(\omega))}, \theta\in[0,2\pi]\}$.

Now, let $(J_{q,s})_{1\leq s \leq L}$  be the Jordan  blocks of  $J_q$. For each block  $J_{q,s}$ associated with  an eigenvalue   $\alpha_s^{(q)}$ of $E_q$ we have  
$$
(z+J_{q,s})^{-1}=\begin{pmatrix} (z+\alpha_s^{(q)})^{-1}&-\frac{(z+\alpha_s^{(q)})^{-2}}{2!}&\cdots&(-1)^{r_s-1}\frac{(z+\alpha_s^{(q)})^{-r_s}}{r_s!}\\
0&(z+\alpha_s^{(q)})^{-1}&\cdots&(-1)^{r_s-2}\frac{(z+\alpha_s^{(q)})^{-r_s+1}}{(r_s-1)!}\\
\vdots &  \ddots&   \ddots& \vdots\\
0&\cdots&(z+\alpha_s^{(q)})^{-1}&-\frac{(z+\alpha_s^{(q)})^{-2}}{2!}\\
0&\cdots&0&(z+\alpha_s^{(q)})^{-1}
\end{pmatrix}
$$
where $1\leq r_s\leq m_{q,s}$.

Then, to estimate the norm of  $(\zeta+J_q)^{-1}$ when $\zeta\in \gamma_j$, it is  sufficient  to evaluate the function $(\zeta+\alpha_s^{(q)})^{-r}$ for $\zeta=-\alpha_j^{(q)}+ \vert \omega \vert^{\delta-1} e^{i(\theta - \arg(\omega))}$ and for the biggest  of the $r_s\leq m_{q,s}$, $1\leq s\leq  L$. One obtain that for every  $\zeta\in \gamma_j$, 
$$
\vert \zeta+ \alpha_s^{(q)} \vert^{-r_s} = \left\{\begin{array}{lll}
\mathcal{O}(1) \; & \text{if} \; s\neq j\\
\mathcal{O}(|\omega|^{m_{q,j}(1-\delta)}) \; & \text{if} \; s = j.
\end{array}\right.
$$
In consequence 
$$
\|(\zeta+J_q)^{-1}\|\leq C |\omega|^{m_{q,j}(1-\delta)}, \quad \forall \zeta\in \gamma_j.
$$ 
Next,
\begin{equation}\label{1.1}
(\zeta+ J_q+wQ_\omega( \zeta))^{-1}=\big(I+w(\zeta+ J_q)^{-1}Q_\omega( \zeta)\big)^{-1}(\zeta+ J_q)^{-1}
\end{equation}
is well defined for  $\zeta$ on the  curve $\gamma_j$ since   $\delta<1+1/m_j$ implies that  $m_j(1-\delta)+1>0$. So \begin{equation}\label{1.2}\|w(\zeta+ J)^{-1}\|=O(|\omega|^{m_j(1-\delta)+1})=o(1), \quad \vert \omega\vert \to 0.\end{equation} 
Consider the difference 
\begin{equation}\label{9sep22}
P_{j,\omega}-\tilde{P}_{j,\omega}=\oint_{\gamma_j}\omega(\zeta+ J_q+wQ_\omega( \zeta))^{-1}( Q_\omega(- \alpha_j)-Q_\omega( \zeta))(\zeta+ J_q+wQ_\omega( -\alpha_j))^{-1}.
\end{equation}

The function $Q_\omega(\zeta)$ is analytic. Then, there exists $C>0$ (independent of $\omega$)  such that $\|Q_\omega(\zeta)-Q_\omega(-{\alpha_j})\|\leq C |\omega \zeta+\omega {\alpha_j}|\leq C |\omega|^\delta$. Thus, using \eqref{1.1},  \eqref{1.2} et \eqref{9sep22} one obtains 
$$
\|P_{j,\omega}-\tilde{P}_{j,\omega}\|
\leq C |\omega|^{2m_j(1-\delta)+2\delta}=o(1).
$$
\end{proof}

Using \eqref{ME} and  \cite[I Lemma 4.10]{kato} we see that  ${\rm rank} {P}_{j,\omega}={\rm rank} \tilde{P}_{j,\omega}$. 
Finally, from \eqref{10}  one obtain that $Q_\omega( -\alpha_j)$ depends analytically  on  $\omega$. Thus,   using  Kato's perturbation theory  \cite[II 1.19]{kato}
$$
{\rm rank} \tilde{P}_{j,\omega}={\rm rank}\oint_{\gamma_j}(z+ J_q+wQ_\omega( -{\alpha_j}))^{-1}dz={\rm rank}\oint_{\gamma_j}(z+ J)^{-1}dz=m_j, 
$$ which implies \eqref{multi}.

To prove the final assertion in Theorem \ref{MTH1} we notice that 
the only difference is to define the operator $\mathcal{L}_\omega(k) $ by 
$$
\mathcal{L}_\omega(k) =   \left( I + \omega {V}_\rho \mathcal{G}(k)\right)^{-1} {V}_\rho \left(a_{-1} \otimes \pi_q+b_{-1}\otimes \pi_p\right)
$$ 
to get  the analogous equation of \eqref{29nov21b} in this case. Then, Lemma \ref{lem2} can be applied with $\Pi=\Pi_{q,p}$ and \eqref{29jan22b} is also obtained. The rest of the proof is similar.

\noindent
\textbf{Acknowledgements}

\noindent
M. Assal and P. Miranda acknowledge the financial support of the project 042133MR-POSTDOC of the Universidad de Santiago de Chile. P. Miranda was partially supported by the chilean fondecyt grant 1201857.

\bibliographystyle{plain}
\bibliography{biblio}

\end{document}